\documentclass[a4paper]{article}
\usepackage[margin=1.25in]{geometry}
\usepackage{authblk}
\usepackage[utf8]{inputenc}
\usepackage{microtype}
\usepackage{amsfonts}
\usepackage{amssymb}
\usepackage{amsmath}
\usepackage{amsthm}
\usepackage{thm-restate}
\usepackage{comment}
\usepackage{xspace}
\usepackage{enumitem}
\usepackage[hidelinks]{hyperref}
\usepackage[capitalize]{cleveref}
\usepackage[ruled,linesnumbered,noend]{algorithm2e}

\crefname{algocf}{Algorithm}{Algorithms}
\newtheorem{theorem}{Theorem}
\newtheorem{lemma}[theorem]{Lemma}

\newtheorem{fact}[theorem]{Fact}

\newtheorem{definition}[theorem]{Definition}
\newtheorem{proposition}[theorem]{Proposition}
\newtheorem{remark}[theorem]{Remark}

\newcommand{\LCP}{\text{\rm LCP}}

\newcommand{\per}{\text{\rm per}}
\newcommand{\lyn}{\text{\rm lyn}}

\newcommand{\dotdot}{\mathinner{\ldotp\ldotp}}
\newcommand{\dd}{\dotdot}

\renewcommand{\P}{\mathsf{P}}

\newcommand{\T}{\mathcal{T}}
\newcommand{\F}{\mathcal{F}}

\newcommand{\Q}{\mathsf{Q}}

\newcommand{\sub}{\subseteq}

\newcommand{\Oh}{\mathcal{O}}

\newcommand{\OhTilde}{\tilde{\mathcal{O}}}
\newcommand{\maxPairLCP}{\mathrm{maxPairLCP}}

\newcommand{\lcs}{\operatorname{lcs}}
\newcommand{\LCSl}{\ensuremath{LCS_{\ell}}\xspace}

\SetKwRepeat{Do}{do}{while}

\newcommand{\defproblem}[3]{
	\vspace{2mm}
	\noindent\fbox{
		\begin{minipage}{0.96\textwidth}
			#1
			
			\smallskip
			\noindent
			{\bf{Input:}} #2
			
			\smallskip
			\noindent
			{\bf{Output:}} #3
		\end{minipage}
	}
	\vspace{2mm}
}

\title{Time-Space Tradeoffs for Finding a Long Common Substring\thanks{Supported by ISF grants no. 1278/16 and 1926/19, a BSF grant no. 2018364, and an ERC grant MPM (no. 683064) under the EU's Horizon 2020 Research and Innovation Programme.}}

\begin{document}
	
\setenumerate{itemsep=0ex, parsep=1pt, topsep=1pt}

\renewcommand\Affilfont{\normalsize}

\author[]{Stav Ben-Nun}
\author[]{Shay Golan}
\author[]{Tomasz Kociumaka}
\author[]{Matan Kraus}

\affil[]{Department of Computer Science, Bar-Ilan University, Ramat Gan, Israel}
\affil[]{\texttt{stav.bennun@live.biu.ac.il, golansh1@cs.biu.ac.il, kociumaka@mimuw.edu.pl, krausma@biu.ac.il}}

\date{\vspace{-1.5cm}}
\maketitle
\begin{abstract}
	We consider the problem of finding, given two documents of total length $n$, a longest string occurring as a substring of both documents. This problem, known as the \textsc{Longest Common Substring (LCS)} problem, has a classic $\Oh(n)$-time solution dating back to the discovery of suffix trees (Weiner, 1973)
	and their efficient construction for integer alphabets (Farach-Colton, 1997).
	However, these solutions require $\Theta(n)$ space, which is prohibitive in many applications. To address this issue, Starikovskaya and Vildhøj (CPM 2013) showed that for $n^{2/3} \le s \le n$, the LCS problem can be solved in $\Oh(s)$ space and $\OhTilde(\frac{n^2}{s})$ time.\footnote{The $\OhTilde$ notation hides $\log^{\Oh(1)} n$ factors.}
	Kociumaka et al. (ESA 2014) generalized this tradeoff to $1 \leq s \leq n$, thus providing a smooth time-space tradeoff from constant to linear space. 
	In this paper, we obtain a significant speed-up for instances where the length $L$ of the sought LCS is large. 
	For $1 \leq s \leq n$, we show that the LCS problem can be solved in $\Oh(s)$ space and $\OhTilde(\frac{n^2}{L\cdot s}+n)$ time. 
	The result is based on techniques originating from the \textsc{LCS with Mismatches} problem (Flouri et al., 2015; Charalampopoulos et al., CPM 2018), on space-efficient locally consistent parsing (Birenzwige et al., SODA 2020), and on the structure of maximal repetitions (runs) in the input documents.
\end{abstract}

\section{Introduction}\label{sec:intro}
The \textsc{Longest Common Substring} (LCS) problem is a fundamental text processing problem with numerous applications; see e.g.~\cite{applicationsAbboud:2015,genomeThomas:2012,biologyFunakoshi:2019}.
Given \emph{two} strings (documents) $S_1,S_2$, the LCS problem asks for a longest string occurring in $S_1$ and $S_2$. We denote the length of the longest common substring by $\lcs(S_1,S_2)$.

The classic text-book solution to the LCS problem is to build the (generalized) suffix tree of the documents and find the node that corresponds to an LCS~\cite{WeinerSTconstruct,Gusfield:97,farachSuffixTree:1997}. 
While this can be achieved in linear time, it comes at the cost of using $\Omega(n)$ words (of $\Theta(\log n)$ bits each) to store the suffix tree.
In applications with large amounts of data or strict space constraints, this renders the classic solution impractical.
To overcome the space challenge of suffix trees, succinct and compressed data structures have been subject to extensive research~\cite{compressedGrossi:2005,compressedNavarro:2007}. Nevertheless, these data structures still use $\Omega(n)$ bits of space in the worst-case.
Starikovskaya and Vildhøj~\cite{diffCover:2013} showed that for $n^{2/3} \le s \le n$, the LCS problem can be solved in $\Oh(\frac{n^2}{s}+s\log n)$ time using $\Oh(s)$ space.
Kociumaka et al.~\cite{KociumakaLastResult} subsequently improved the running time to $\Oh(\frac{n^2}{s})$
and extended the parameter range to $1\le s \le n$.

These previous works also considered a generalized version of the LCS problem,
where the input consists of $m$ documents $S_1,S_2,\ldots,S_m$ (still of total length $n$) and an integer $2 \leq d \leq m$. The task there is to compute a longest string occurring as a substring of at least $d$ of the $m$ input documents. In this setting, Starikovskaya and Vildhøj~\cite{diffCover:2013} 
achieve $\Oh(\frac{n^2 \log^2 n}{s}(d\log^2{n}+d^2))$ time and $\Oh(s)$ space for $n^{2/3} \le s \le n$,
whereas Kociumaka et al.~\cite{KociumakaLastResult} showed a solution which takes  $\Oh(\frac{n^2}{s})$ time and $\Oh(s)$ space for $1 \leq s \leq n$. The cost of this algorithm matches both a classic $\Theta(n)$-space algorithm~\cite{Hui:92} and the time-space tradeoff for $d=m=2$.
Nevertheless, in this paper we focus on the LCS problem for \emph{two} strings only.

Kociumaka et al.~\cite{KociumakaLastResult} additionally provided a lower bound which states that any deterministic algorithm using $s \le \frac{n}{\log n}$ space must cost $\Omega(n\sqrt{\log(n/(s \log n))/ \log\log(n/(s\log n))})$ time. This lower bound is actually derived for the problem of distinguishing whether $\lcs(S_1,S_2)=0$, i.e., deciding if the two input strings have any character in common.
This state of affairs naturally leads to a question of whether distinguishing between $\lcs(S_1,S_2)<\ell$ and $\lcs(S_1,S_2)\ge \ell$ gets easier as $\ell$ increases, or equivalently,
whether $L:=\lcs(S_1,S_2)$ can be computed more efficiently when $L$ is large.
This case is relevant for applications since the existence of short common substrings is less meaningful for measuring string similarity.

\subsection{Our Results}
We provide new sublinear-space algorithms for the LCS problem optimized for inputs with a long common substring. 
The algorithms are designed for the word-RAM model with word size $w=\Theta(\log n)$, and they work for integer alphabets $\Sigma = \{1,2,\ldots,n^{\Oh(1)}\}$. 
Throughout the paper, the input strings reside in a read-only memory and any space used by the algorithms is a working space; furthermore, we represent the output by witness occurrences in the input strings so that it fits in $\Oh(1)$ machine words. 
Our main result is as follows:
\begin{restatable}{theorem}{thmmainresult}\label{thm:mainresult}
	Given $s$ with $1\le s \le n$, the LCS problem with $L=\lcs(S_1,S_2)$ can be solved deterministically in $\Oh(s)$ space and $\Oh(\frac{n^2\log n \log^*n }{s \cdot L}+n \log n)$ time,%
	\footnote{The \emph{iterated logarithm} function $\log^*$ is formally defined with $\log^* x = 0$ for $x\le 1$ and $\log^*x = 1+\log^*(\log x)$ for $x>1$. 
	In other words $\log^*n$ is the smallest integer $k\ge 0$ such that $\log^{(k)}x\le 1$.}
	and in $\Oh(s)$ space and $\Oh(\frac{n^2\log n}{s \cdot L}+n \log n)$ time with high probability using a Las-Vegas randomized algorithm.
\end{restatable}
We remark that \cref{thm:mainresult} improves upon the result of Kociumaka et al.~\cite{KociumakaLastResult} whenever
$s < \frac{n}{\log n}$ and $L > \log n \log^* n$ (or $L > \log n$ if randomization is allowed).

We also show that the $\log$ factors be removed from the running times in \cref{thm:mainresult} if $s=\Theta(1)$. In fact, this yields an improvement upon \cref{thm:mainresult} as long as $s < \log n \log^*n$.
\begin{restatable}{theorem}{thmconstantspaceresult}\label{thm:constantspaceresult}
	The LCS problem can be solved deterministically in $\Oh(1)$ space and $\Oh(\frac{n^2}{L})$ time,
	where $L=\lcs(S_1,S_2)$.
\end{restatable}

As a step towards our main result, we solve the \LCSl problem defined below.

\defproblem{\textsc{Longest Common Substring with Threshold $\ell$} (\LCSl)}{%
	Two strings $S_1$ and $S_2$ (of length at most $n$), an integer threshold $\ell\in [n]$}{
	A common substring $G$ of $S_1$ and $S_2$ such that:
	\begin{enumerate}
		\item if $\ell \le \lcs(S_1,S_2)< 2\ell$, then $|G| = \lcs(S_1,S_2)$;
		\item if $\lcs(S_1,S_2)\ge 2\ell$, then $|G|\ge 2\ell$.
	\end{enumerate}
}

If $\lcs(S_1,S_2)<\ell$, then \LCSl allows for an arbitrary common substring in the output.
\begin{remark}\label{rem:equiv}
	Note the following equivalent characterization of the output $G$ of \LCSl:
	for every common substring $T$ with $\ell \le |T| \le 2\ell$,
	the common substring $G$ is of length $|G|\ge |T|$.
\end{remark}

\begin{restatable}{theorem}{thmbaseresult}\label{thm:base_improved}
	The \LCSl problem can be solved deterministically in $\Oh(\frac{n\log^* n}{\ell})$ space and $\Oh(n \log n)$ time, and in $\Oh(\frac{n}{\ell})$ space and $\Oh(n\log n)$ time with high probability using a Las-Vegas randomized algorithm.
\end{restatable}

\subsection{Related work}
The LCS problem has been studied in many other settings. 
Babenko and Starikovskaya~\cite{DBLP:journals/poit/BabenkoS11}, Flouri et al.~\cite{combinedTreeFlouri:2015}, Thankachan et al.~\cite{DBLP:journals/jcb/ThankachanAA16}, and Kociumaka et al.~\cite{approxMismatchKociumaka:2019}  studied the LCS with $k$ mismatches problem, where the occurrences in $S_1$ and $S_2$ can be at Hamming distance up to $k$. Charalampopoulos et al.~\cite{charalampopoulos:2018} showed that this problem becomes easier when the strings have a long common substring with $k$ mismatches (similarly to what we obtain for LCS with no mismatches).
Thankachan et al.~\cite{DBLP:conf/recomb/ThankachanACA18} and Ayad et al.~\cite{errorsAyad:2018} considered a related problem with edit distance instead of the Hamming distance. Alzamel et al.~\cite{lyndonAlzamel:2019}
proposed an $\OhTilde(n)$-time algorithm for the Longest Common Circular Substring problem, where occurrences in $S_1$ and $S_2$ can be cyclic rotations of each other.

Amir et al.~\cite{DBLP:conf/spire/AmirCIPR17} studied the problem of answering queries asking for the LCS after a single edit in either of the two original input strings. Subsequently, Amir et al.~\cite{fullyDynamicAmir:2019} 
and Charalampopoulos et al.~\cite{DBLP:conf/icalp/CharalampopoulosGP20} considered a fully dynamic version of the problem, in which the edit operations are applied sequentially, ultimately achieving $\OhTilde(1)$ time per operation.

\subsection{Algorithmic Overview}
We first give an overview of the algorithm of \cref{thm:base_improved}.
Then, we derive \cref{thm:mainresult} in two steps, with an $\Oh(s)$-space solution to the \LCSl problem as an intermediate result.

\subparagraph{An $\OhTilde(n/\ell)$-space algorithm for the \LCSl problem.} 
In \cref{sec:anchored}, we define an \emph{anchored} variant of the \textsc{Longest Common Substring} problem (LCAS). 
In the LCAS problem, we are given two strings $S_1$, $S_2$ and sets of positions $A_1$ and $A_2$, and we wish to find a longest common substring which can be obtained by extending (to the left and to the right) $S_1[p_1]$ and $S_2[p_2]$ for some $(p_1,p_2)\in A_1 \times A_2$.
We then reduce the LCAS problem to the \textsc{Two String Families LCP} problem, introduced by Charalampopoulos et al.~\cite{charalampopoulos:2018} in the context of finding LCS with mismatches.

In \cref{sec:choosingAnchors}, we show how to solve the \LCSl problem by selecting positions in $A_1$ and $A_2$ so that \emph{every} common substring $T$ of $S_1$ and $S_2$ with $|T| \ge \ell$ can be obtained by extending $S_1[p_1]$ and $S_2[p_2]$ for some $(p_1,p_2)\in A_1 \times A_2$.
To make this selection, we use \emph{partitioning sets} by Birenzwige et al.~\cite{locallyConsistentBGP:2018}, which consist of $\OhTilde(\frac{n}{\ell})$ positions chosen in a locally consistent manner. However, since partitioning sets do not select positions in long periodic regions, our algorithms use \emph{maximal repetitions} (runs)~\cite{runsKolpakov:1999,runsBannai:2017} and their \emph{Lyndon roots}~\cite{lyndonCrochemore:2012} to add $\Oh(\frac{n}{\ell})$ extra positions. Overall, we get an $\OhTilde(\frac{n}{\ell})$-space and $\OhTilde(n)$-time algorithm for the \LCSl problem.

\subparagraph{An $\Oh(s)$ space algorithm for the \LCSl problem.} 
In \cref{sec:knownLenAlgo}, we give a time-space tradeoff for the \LCSl problem.
The algorithm partitions the input strings into overlapping substrings, executes the algorithm of \cref{sec:choosingAnchors} for each pair of substrings, and returns the longest among the common substrings obtained from these calls. For a tradeoff parameter $1\le s\le n$, the algorithm takes $\Oh(s)$ space and $\OhTilde(\frac{n^2}{s \cdot \ell} + n)$ time.

\subparagraph{A solution to the LCS problem.} In \cref{sec:finalAlgo}, we show how to search for LCS by repeatedly solving the \LCSl problem with different choices of $\ell$. We get an algorithm that takes $\Oh(s)$ space and $\OhTilde(\frac{n^2}{s \cdot L} + n)$ time, where $L=\lcs(S_1,S_2)$, as stated in \cref{thm:mainresult}.

\section{Preliminaries}\label{sec:preliminaries}
For $1\leq i<j\leq n$, denote the integer \emph{intervals} $[i\dd j] = \{i,i+1,\dots,j\}$ and $[k]=\{1,2,\dots,k\}$.

A string $S$ of length $n=|S|$ is a finite sequence of characters $S[1]S[2]\cdots S[n]$ over an alphabet $\Sigma
$; in this paper, we consider polynomially-bounded integer alphabets $\Sigma=[1\dd n^{\Oh(1)}]$.
The string $S^r = S[n]S[n-1]\cdots S[1]$ is called the \emph{reverse} of the string $S$.

A string $T$ is a \emph{substring} of a string $S$ if $T=S[x]S[x+1]\cdots S[y]$
for some $1\le x \le y \le |S|$. We then say that $T$ \emph{occurs} in $S$ at position $x$,
and we denote the \emph{occurrence} by $S[x\dd y]$.
We call $S[x\dd y]$ a \emph{fragment} of $S$.
A fragment $S[x\dd y]$ is a \emph{prefix} of $S$ if $x=1$ and a \emph{suffix} of $S$ if $y=|S|$.
These special fragments are also denoted by $S[\dd y]$ and $S[x\dd]$, respectively.
A \emph{proper fragment} of $S$ is any fragment other than $S[1\dd |S|]$.
A common prefix (suffix) of two strings $S_1,S_2$ is a string that occurs as a prefix (resp. suffix)
of both $S_1$ and $S_2$.
The longest common prefix of $S_1$ and $S_2$ is denoted by $\LCP(S_1,S_2)$, and the longest common suffix is denoted by $\LCP^r(S_1,S_2)$. Note that $\LCP^r(S_1,S_2) = (LCP(S_1^r,S_2^r))^r$.

An integer $k\in [|S|]$, is a \emph{period} of a string $S$ if $S[i]=S[i+k]$ for $i\in [|S|-k]$. 
The shortest period of $S$ is denoted by $\per(S)$. 
If $\per(S) \leq \frac12|S|$, we say that $S$ is \textit{periodic}.
A periodic fragment $S[i\dd j]$ is called a \emph{run}~\cite{runsKolpakov:1999,runsBannai:2017} if it cannot be extended (to the left nor to the right)
without increasing the shortest period.
For a pair of parameters $d$ and $\rho$, we say that a run $S[i\dd j]$ is a \emph{$(d,\rho)$-run} if $|S[i\dd j]|\ge d$ and $\per(S[i\dd j])\le \rho$. Note that every periodic fragment $S[i'\dd j']$ with $|S[i'\dd j']|\ge d$ and $\per(S[i'\dd j'])\le \rho$ can be uniquely extended to a $(d,\rho)$-run $S[i\dd j]$ while preserving the shortest period $\per(S[i\dd j])=\per(S[i'\dd j'])$.

\subparagraph{Tries and suffix trees.} Given a set of strings $\F$, the compact trie~\cite{morrison:1968} of these strings is the tree obtained by compressing each path of nodes of degree one in the trie~\cite{briandais:1959,fredkin:1960} of the strings in $\F$, which takes $\Oh(|\F|)$ space. Each edge in the compact trie has a label represented as a fragment of a string in $\F$. The suffix tree~\cite{WeinerSTconstruct} of a string $S$ is the compact trie of all the suffixes of $S$.
The sparse suffix tree~\cite{sstKarkkainen:1996,DBLP:journals/talg/Bille0GKSV16,DBLP:conf/stacs/IKK14,DBLP:conf/soda/GawrychowskiK17} of a string $S$ is the compact trie of selected suffixes $\{S[i\dd ] : i \in B\}$
specified by a set of positions $B\sub [|S|]$.

\section{Longest Common Anchored Substring problem}\label{sec:anchored}
In this section, we consider an \emph{anchored} variant of the \textsc{Longest Common Substring} problem.
Let $A_1$ and $A_2$ be sets of distinguished positions,  called \emph{anchors}, in strings $S_1$ and $S_2$, respectively.
We say that a string $T$ is a \emph{common anchored substring} of $S_1$ and $S_2$ with respect to $A_1$ and $A_2$ if it has occurrences $S_1[i_1\dd j_1]=T=S_2[i_2\dd j_2]$ with a \emph{synchronized pair of anchors}, i.e., with some anchors $p_1\in A_1$ and $p_2\in A_2$
such that $p_1 - i_1 = p_2 - i_2 \in [0,|T|]$.\footnote{Note that the anchors could be at positions $p_1=j_1+1$ and $p_2 = j_2+1$ (if $p_1-i_1 = p_2-i_2 = |T|$).}

\defproblem{\textsc{Longest Common Anchored Substring (LCAS)}}{%
	Two strings $S_1$, $S_2$ (of length at most $n$) and two sets of anchors $A_1\sub [|S_1|]$, $A_2\sub [|S_2|]$.
}{%
	A longest common anchored substring of $S_1$ and $S_2$ with respect to $A_1,A_2$.
}

We utilize the following characterization of the longest common anchored substring.
\begin{fact}\label{fct:lcas}
	The length of a longest common anchored substring of two strings $S_1$ and $S_2$ (with respect to anchors at $A_1$ and $A_2$) is
	\[
	\max\{|\LCP(S_1[p_1\dd],S_2[p_2\dd])|+|\LCP^r(S_1[\dd p_1-1],S_2[\dd p_2-1])| :  p_1\in A_1, p_2\in A_2\}.
	\]
	Moreover, for any $(p_1,p_2)\in A_1\times A_2$ maximizing this expression, a longest common anchored substring
	is the concatenation $\LCP^r(S_1[\dd p_1-1],S_2[\dd p_2-1])\cdot \LCP(S_1[p_1\dd],S_2[p_2\dd])$.
\end{fact}

The characterization of \cref{fct:lcas} lets us use the following problem, originally defined
in a context of computing LCS with mismatches.

\defproblem{\textsc{Two String Families LCP} (Charalampopoulos et al.~\cite{charalampopoulos:2018})}{%
	A compact trie $\T(\F)$ of a family of strings $\F$ and two sets $\P,\Q\sub \F^2$.
}{%
	The value $\maxPairLCP(\P,\Q)$, defined as
	\[\maxPairLCP(\P,\Q)=\max\{|\LCP(P_1,Q_1)|+|\LCP(P_2,Q_2)| : (P_1,P_2)\in \P \text{, }(Q_1,Q_2)\in \Q\},\]
	along with pairs $(P_1,P_2)\in \P$ and $(Q_1,Q_2)\in \Q$ for which the maximum is attained.
}

Charalampopoulos et al.~\cite{charalampopoulos:2018} observed that the \textsc{Two String Families LCP} problem
can be solved using an approach by Crochemore et al.~\cite{combinedTreeCrochemore:2006} and Flouri et al.~\cite{combinedTreeFlouri:2015}.

\begin{lemma}[{\cite[Lemma 3]{charalampopoulos:2018}}]\label{lem:twoStringLCP}
	The \textsc{Two String Families LCP} problem can be solved in $\Oh(|\F|+N\log N)$ time using $\Oh(|\F|+N)$ space%
	\footnote{The original formulation of~\cite[Lemma 3]{charalampopoulos:2018} does not discuss the space complexity.
	However, an inspection of the underlying algorithm, described in~\cite{combinedTreeCrochemore:2006,combinedTreeFlouri:2015}, easily yields this additional claim.}, where $N=|\P|+|\Q|$.
\end{lemma}

By \cref{fct:lcas}, the \textsc{LCAS} problem reduces to the \textsc{Two String Families LCP} problem~with:
\begin{align}
\begin{split}\label{eq:lcasreductF}
\F&=\{S_1[p\dd ] : p \in A_1\} \cup \{(S_1[\dd p-1])^r : p \in A_1\} \\ &\qquad \cup \{S_2[p\dd ] : p \in A_2\} \cup \{(S_2[\dd p-1])^r : p\in A_2\},
\end{split}\\
\P&=\{(S_1[p\dd ],(S_1[\dd p-1])^r) : p \in A_1\},\label{eq:lcasreductP}\\
\Q&=\{(S_2[p\dd ],(S_2[\dd p-1])^r) : p \in A_2\}.\label{eq:lcasreductQ}
\end{align}

The following theorem provides an efficient implementation of this reduction.
The most challenging step, to construct the compacted trie $\T(\F)$, is delegated to the work of Birenzwige et al.~\cite{locallyConsistentBGP:2018}, who show that a sparse suffix tree of a length-$n$ string $S$ with $B\subseteq[n]$ can be constructed deterministically in $\Oh(n\log\frac{n}{|B|})$ time and $\Oh(|B| + \log n)$ space.
\begin{theorem}\label{thm:lcas}
	The \textsc{Longest Common Anchored Substring} problem can be solved in $\Oh(n\log n)$ time using $\Oh(|A_1|+|A_2| + \log n)$ space.
\end{theorem}
\begin{proof}
	We implicitly create a string $S=S_1\$S_1^r\$S_2\$S_2^r$ and construct a sparse suffix tree of $S$
	containing the following suffixes: $S_1[p\dd ]\$S_1^r\$S_2\$S_2^r$ and
	$(S_1[\dd p-1])^r\$ S_2\$ S_2^r$ for $p\in A_1$, as well as $S_2[p\dd ]\$ S_2^r$ and
	$(S_2[\dd p-1])^r$ for $p\in A_2$. We then trim this tree, cutting edges immediately above any $\$$ on their labels, which results in the compacted trie $\T(\F)$ for the family $\F$ defined in~\eqref{eq:lcasreductF}.	
	We then build $\P$ and $\Q$ according to~\eqref{eq:lcasreductP} and~\eqref{eq:lcasreductQ}, respectively,
	and solve an instance of the \textsc{Two String Families LCP} problem specified by $\T(\F),\P,\Q$.
	This yields pairs in $\P$ and $\Q$ for which $\maxPairLCP(\P,\Q)$ is attained.
	We retrieve the underlying indices $p_1\in A_1$ and $p_2\in A_2$ and derive a longest common anchored substring of $S_1$ and $S_2$ according to \cref{fct:lcas}: $\LCP^r(S_1[\dd p_1-1],S_2[\dd p_2-1])\cdot \LCP(S_1[p_1\dd],S_2[p_2\dd])$.
	
	With the sparse suffix tree construction of \cite{locallyConsistentBGP:2018}
	and the algorithm of \cref{lem:twoStringLCP} that solves the \textsc{Two String Families LCP} problem,
	the overall running time is $\Oh(n\log \frac{n}{N} + N\log N)=\Oh(n \log n)$ and the space complexity is $\Oh(N + \log n)$, where $N=|A_1|+|A_2|$. 
\end{proof}

\section{Space-efficient $\OhTilde(n)$-time algorithm for the \LCSl problem}\label{sec:choosingAnchors}
Our approach to solve the \LCSl problem is via a reduction to the \textsc{LCAS} problem.
For this, we wish to select anchors $A_1\sub [|S_1|]$ and $A_2\sub [|S_2|]$ so that every common substring $T$ of length at least $\ell$ is a common anchored substring. In other words, we need to make sure that $T$ admits occurrences $S_1[i_1\dd j_1]=T=S_2[i_2\dd j_2]$ with a synchronized pair of anchors.

As a warm-up, we describe a simple selection of $\Oh(n/\sqrt{\ell})$ anchors based on \emph{difference covers}~\cite{DBLP:journals/tocs/Maekawa85}, which have already been used by Starikovskaya and Vildh{\o}j~\cite{diffCover:2013} in a time-space tradeoff for the LCS problem. For every two integers $\tau,m$ with $1\le \tau \le m$,
this technique yields a set $DC_\tau(m)\sub [m]$ of size $\Oh(m/\sqrt{\tau})$ such that
for every two indices $i_1,i_2\in [m-\tau+1]$, there is a shift $\Delta\in [0\dd \tau-1]$
such that $i_1+\Delta$ and $i_2+\Delta$ both belong to $DC_\tau(m)$.
Hence, to make sure that every common substring of length at least $\ell$ is anchored,
it suffices to select all $\Oh(n/\sqrt{\ell})$ positions in $DC_{\ell}(n)$ as anchors: $A_1=[|S_1|]\cap DC_{\ell}(n)$ and $A_2 = [|S_2|]\cap DC_{\ell}(n)$.

We remark that such selection of anchors is \emph{non-adaptive}: it does not depend on contents of the strings $S_1$ and $S_2$, but only on the lengths of these strings (and the parameter $\ell$).
In fact, any non-adaptive construction needs $\Omega(n/\sqrt{\ell})$ anchors in order to guarantee that every common substring $T$ of length at least $\ell$ is a common anchored substring.
In the following, we show how adaptivity allows us to achieve the same goal using only $\OhTilde(n/\ell)$ anchors.

\subsection{Selection of Anchors: the non-periodic case}\label{sec:synchronized_position}
We first show how to accommodate
common substrings $T$ of length $|T|\ge \ell$ that do not contain a $(\frac35\ell,\frac15\ell)$-run.
The idea is to use \emph{partitioning sets} by Birenzwige et al.~\cite{locallyConsistentBGP:2018}.
\begin{definition}[Birenzwige et al.~\cite{locallyConsistentBGP:2018}]\label{def:partitioning_set}
	A set of positions $P\subseteq [n]$ is called a \emph{$(\tau,\delta)$-partitioning set} of a length-$n$ string $S$, for some parameters $\tau,\delta \in [n]$, if it has the following properties:
	\begin{description}
		\item[Local Consistency:]
		For every two indices $i,j\in[1+\delta\dd n-\delta]$ such that $S[i-\delta\dd i+\delta]=\linebreak S[j-\delta\dd j+\delta]$, we have $i\in P$ if and only if $j\in P$.
		\item[Compactness:]
		If $p_i<p_{i+1}$ are two consecutive positions in $P\cup\{1,n+1\}$ such that $p_{i+1} > p_i + \tau$,
		then $u=S[p_i\dd p_{i+1}-1]$ is periodic with period $\per(u)\le\tau$. 
	\end{description}
\end{definition}

Note that any $(\tau,\delta)$-partitioning set is also a $(\tau',\delta')$-partitioning set for any $\tau'\ge \tau$ and $\delta'\ge \delta$. 
The selection of anchors is based on an arbitrary $(\frac15\ell,\frac15\ell)$-partitioning set $P$ of the string $S_1S_2$: for every position $p \in P$,
$p$ is added to $A_1$ (if $p\le |S_1|$) or $p-|S_1|$ is added to $A_2$ (otherwise).
Below, we show that this selection satisfies the advertised property.%
\begin{lemma}\label{lemma:findLCSifnotPeriodic}
	Let $T$ be a common substring of length $|T| \geq \ell$
	which does not contain a $(\frac35\ell,\frac15\ell)$-run. Then, $T$ is a common anchored substring with respect to $A_1,A_2$ defined above.
\end{lemma}
\begin{proof}
	Let $S_1[i_1\dd j_1]$ and $S_2[i_2\dd j_2]$ be arbitrary occurrences of $T$ in $S_1$ and $S_2$,
	respectively.
	If there is a position $p_1\in A_1$ with $p_1 \in [i_1+\delta \dd  j_1-\delta]$,
	then the position $p_2 = i_2 + (p_1-i_1)$ belongs to $A_2$ by the local consistency property of the underlying partitioning set, due to $S_1[p_1-\delta\dd p_1+\delta]=S_2[p_2-\delta\dd p_2+\delta]$.
	Hence, $(p_1,p_2)$ is a synchronized pair of anchors and $T$ is a common anchored substring with respect to $A_1,A_2$.
	
	If there is no such position $p_1\in A_1$, then $S_1[i_1+\delta\dd j_1-\delta]$ is contained within a block between two consecutive positions of the partitioning set. The length of this block is at least $|T|-2\delta \geq \frac35\ell > \tau$,	so the block is periodic by the compactness property of the partitioning set. Hence, $\per(T[1+\delta\dd |T|-\delta])=\per(S_1[i_1+\delta\dd j_1-\delta])\le \tau\le \frac15\ell$.
	A $(\frac35\ell,\frac15\ell)$-run in $T$ can thus be obtained by maximally extending $T[1+\delta\dd |T|-\delta]$ without increasing the shortest period. Such a run in $T$ is a contradiction to the assumption.
\end{proof}

Birenzwige et al.~\cite{locallyConsistentBGP:2018} gave a deterministic algorithm that constructs a $(\tau,\tau\log^*n)$-partitioning set of size $\Oh(\frac{n}{\tau})$ in $\Oh(n\log\tau)$ time using $\Oh(\frac {n}{\tau} + \log \tau)$ space. Setting appropriate $\tau = \Theta(\ell / \log^*n)$, we get an $(\frac15\ell,\frac15\ell)$-partitioning set of size $\Oh(\frac{n\log^* n}{\ell})$.

Furthermore, Birenzwige et al.~\cite{locallyConsistentBGP:2018} gave a Las-Vegas randomized algorithm that constructs a $(\tau,\tau)$-partitioning set of size $\Oh(\frac{n}{\tau})$ in $\Oh(n + \tau \log^2 n)$ time with high probability, using $\Oh(\frac{n}{\tau} + \log n)$ space. Setting $\tau = \frac15\ell$, we get an $(\frac15\ell,\frac15\ell)$-partitioning set of size $\Oh(\frac{n}{\ell})$.

\subsection{Selection of anchors: the periodic case}\label{sec:periodic_case}
In this section, for any parameters $d,\rho\in [n]$ satisfying $d\ge 3\rho-1$, we show how to accommodate all common substrings containing a $(d,\rho)$-run by selecting $\Oh(\frac{n}{d})$ anchors.
This method is then used for $d=\frac35\ell$ and $\rho=\frac15\ell$ to complement the selection in \cref{sec:synchronized_position}.

Let $T$ be a common substring of $S_1$ and $S_2$ containing a $(d,\rho)$-run. 
We consider two cases depending on whether the run is a proper fragment of $T$ or the whole $T$.
In the first case, it suffices to select as anchors the first and the last position of every $(d,\rho)$-run.
\begin{lemma}\label{lemma:addEdgeOfPeriodic}
	Let $A_1$ and $A_2$ contain the boundary positions of every $(d,\rho)$-run in $S_1$ and $S_2$, respectively.
	If $T$ is a common substring of $S_1$ and $S_2$ with a $(d,\rho)$-run $r$ as a proper fragment,
	then $T$ is a common anchored substring of $S_1$ and $S_2$ with respect to $A_1,A_2$.
\end{lemma}
\begin{proof}
	In the proof, we assume that $r=T[i\dd j]$ with $i\ne 1$. The case of $j\ne |T|$ is symmetric.
	
	Suppose that an occurrence of $T$ in $S_1$ starts at position $i_1$.
	The fragment matching $r$, i.e, $S_1[i_1+i-1\dd i_1+j-1]$, is periodic, has length at least $d$ and period at most $\rho$, so it can be extended to a $(d,\rho)$-run in $S_1$.
	This run in $S_1$ starts at position $i_1+i-1$ due to $T[i-1]\ne T[i+\per(r)-1]$, so $p_1 := i_1+i-1\in A_1$.
	The same argument shows that $p_2:=i_2+i-1\in A_2$ if $T$ occurs in $S_2$ at position $i_2$.
	Hence, $(p_1,p_2)$ is a synchronized pair of anchors and $T$ is a common anchored substring with respect to $A_1,A_2$.
\end{proof}

We are left with handling the case when the whole $T$ is a $(d,\rho)$-run, i.e., 
when $T$ is periodic with $|T|\ge d$ and $\per(T)\le \rho$.
In this case, we cannot guarantee that \emph{every} pair of occurrences of $T$ in $A_1$ and $A_2$
has a synchronized pair of anchors. For example, if $T=\texttt{a}^d$ and $S_1=S_2=\texttt{a}^n$ with $n\ge 2d$,
this would require $\Omega(n/\sqrt{d})$ anchors. (There are $\Omega(n^2)$ pairs of occurrences,
and each pair of anchors can accommodate at most $d+1$ out of these pairs.)

Hence, we focus on the leftmost occurrences of $T$ and observe that they start within the first $\per(T)$
positions of $(d,\rho)$-runs. To achieve synchronization in these regions,
we utilize the notion of the \emph{Lyndon root}~\cite{lyndonCrochemore:2012} $\lyn(X)$ of a periodic string $X$, defined as the lexicographically smallest rotation of $X[1\dd \per(X)]$. For each $(d,\rho)$-run $x$, we select as anchors the leftmost two positions where $\lyn(x)$ occurs within $x$ (they must exist due to $d\ge 3\rho-1$).

\begin{lemma}\label{lemma:addLyndonRoots}
	Let $A_1$ and $A_2$ contain the first two positions where the Lyndon root occurs within each $(d,\rho)$-run
	of $S_1$ and $S_2$, respectively. 
	If $T$ is a common substring of $S_1$ and $S_2$ such that the whole $T$ is a $(d,\rho)$-run,
	then $T$ is a common anchored substring of $S_1$ and $S_2$.
\end{lemma}
\begin{proof}
	Let $k$ be the leftmost position where $\lyn(T)$ occurs in $T$ and $T=S_1[i_1\dd j_1]$ be the leftmost occurrence of $T$ in $S_1$.
	Since $T$ is a $(d,\rho)$-run, $S_1[i_1\dd j_1]$ can be extended to a $(d,\rho)$-run $x$ in $S_1$.
	Note that $S_1[i_1\dd j_1]$ starts within the first $\per(T)$ positions of $x$;
	otherwise, $T$ would also occur at position $i_1-\per(T)$.
	Consequently, position $i_1+k-1$ is among the first $2\per(T)$ positions of $x$,
	and it is a starting position of $\lyn(x)=\lyn(T)$.
	As the subsequent occurrences of $\lyn(x)$ within $x$	start $\per(T)$ positions apart, we conclude that  $i_1+k-1$ is one of the first two positions
	where $\lyn(x)$ occurs within $x$. Thus, $p_1 := i_1+k-1\in A_1$.
	Symmetrically, $p_2 := i_2+k-1$ is added to $A_2$. Hence, $(p_1,p_2)$ is a synchronized pair of anchors and $T$ is a common anchored substring with respect to $A_1,A_2$.
\end{proof}

It remains to prove that \cref{lemma:addEdgeOfPeriodic,lemma:addLyndonRoots} yield $\Oh(\frac{n}{d})$
anchors and that this selection can be implemented efficiently. 
We use the following procedure as a subroutine:

\begin{lemma}[Kociumaka et al.~{\cite[Lemma 6]{shortestPeriod:2015}}]\label{lemma:findShortestPeriodic}
	Given a string $S$, one can decide in $\Oh(|S|)$ time and $\Oh(1)$ space if $S$ is periodic and, if so, compute $per(S)$.
\end{lemma}

First, we bound the number of $(d,\rho)$-runs and explain how to generate them efficiently.

\begin{lemma}\label{lemma:findAllPeriodicSubstrings}
	Consider a string $S$ of length $n$ and positive integers $\rho,d$ with $3\rho-1\le d \le n$.
	The number of $(d,\rho)$-runs in $S$ is $\Oh(\frac{n}{d})$.
	Moreover, there is an $\Oh(1)$-space $\Oh(n)$-time deterministic algorithm reporting them one by one along with their periods.
\end{lemma}
\begin{proof}
	Consider all fragments $u_k=S[k\rho \dd (k+2)\rho-1]$ with boundaries within $[n]$.
	Observe that each $(d,\rho)$-run $v$ contains at least one of the fragments $u_k$:
	if $v=S[i\dd j]$, then $u_k$ with $k=\lceil i/\rho\rceil$ starts at $k\rho \ge i$ and ends  
	at $(k+2)\rho-1\le i+\rho-1+2\rho-1=i+3\rho-2 \le i+d-1 \le j$.
	Moreover, if $v$ contains $u_k$, then $u_k$ is periodic with period $\per(u_k)=\per(v)\le \rho=\frac12|u_k|$ (the first equality is due to $|u_k|=2\rho \ge 2\per(v)$ and the periodicity lemma~\cite{fine1965uniqueness}),
	and $v$ can be obtained by maximally extending $u_k$ without increasing the shortest period.
	
	This leads to a simple algorithm generating all $(d,\rho)$-runs,
	which processes subsequent integers $k$ as follows:
	First, apply \cref{lemma:findShortestPeriodic} to test if $u_k$ is periodic 
	and retrieve its period $\rho_k$.
	If this test is successful,
	then maximally extend $u_k$ while preserving the period $\rho_k$,
	and denote the resulting fragment by $v_k$.
	If $|v_k|\ge d$, then report $v_k$ as a $(d,\rho)$-run.
	We also introduce the following optimization: after processing $k$,
	skip all indices $k'>k$ for which $u_{k'}$ is still contained in $v_k$.
	(These indices $k'$ are irrelevant due to $v_{k'}=v_k$ and they form an integer interval.)
	
	The algorithm of \cref{lemma:findShortestPeriodic} takes constant space and $\Oh(|u_k|)$ time,
	which sums up to $\Oh(n)$ across all indices $k$.
	The naive extension of $u_k$ to $v_k$ takes constant space and $\Oh(|v_k|)$ time.
	Due to the optimization, no two explicitly generated extensions $v_k$ contain the same fragment $u_{k'}$.
	Hence, the total length of the fragments $v_k$ (across indices $k$ which were not skipped) is $\Oh(n)$.
	Thus, the overall running time is $\Oh(n)$ and the number of runs reported is $\Oh(\frac{n}{d})$.
\end{proof}

We conclude with a complete procedure generating anchors in the periodic case.
\begin{proposition}\label{proposition:periodic-anchored-sets}
	There exists an $\Oh(1)$-space $\Oh(n)$-time algorithm that, given two strings $S_1,S_2$ of total length $n$, and parameters $d,\rho\in [n]$ with $d\ge 3\rho-1$,
	outputs sets $A_1,A_2$ of size $\Oh(\frac{n}{d})$ satisfying the following property:
	If $T$ is a common substring of $S_1$ and $S_2$ containing a $(d,\rho)$-run, then $T$ is a common anchored substring of $S_1$ and $S_2$ with respect to $A_1,A_2$.
\end{proposition}
\begin{proof}
	The algorithm first uses the procedure of \cref{lemma:findAllPeriodicSubstrings} to retrieve
	all $(d,\rho)$-runs in $S_1$ along with their periods. For each $(d,\rho)$-run $S_1[i\dd j]$, Duval's algorithm~\cite{duval:1982} is applied to find the minimum cyclic rotation of $S_1[i\dd i+\per(S_1[i\dd j])-1]$ in order to determine the Lyndon root $\lyn(S_1[i\dd j])$ represented by its occurrence at position $i+\Delta$ of $S_1$. Positions $i$, $i+\Delta$, $i+2\Delta$, and $j$ are reported as anchors in $A_1$.
	The same procedure is repeated for $S_2$ resulting in the elements of $A_2$ being reported one by one.
	
	The space complexity of this algorithm is $\Oh(1)$, and the running time is $\Oh(n)$ (for \cref{lemma:findAllPeriodicSubstrings}) plus $\Oh(\rho)=\Oh(d)$ per $(d,\rho)$-run (for Duval's algorithm). 
	This sums up to $\Oh(n)$ as the number of $(d,\rho)$-runs is $\Oh(\frac{n}{d})$.
	For the same reason, the number of anchors is $\Oh(\frac{n}{d})$.
	
	For each $T$, the anchors satisfy the required property due to \cref{lemma:addEdgeOfPeriodic} or \cref{lemma:addLyndonRoots}, depending on whether the $(d,\rho)$ run contained in $T$ is a proper fragment of $T$ or not.
\end{proof}

\subsection{$\OhTilde(n/\ell)$-space algorithm for arbitrary $\ell$}


The \LCSl problem reduces to an instance of the \textsc{LCAS} problem 
with a combination of anchors for the non-periodic case and the periodic case.
This yields the following result:
\begin{theorem}\label{thm:base_simple}
	The \LCSl problem can be solved deterministically in $\Oh(\frac{n\log^* n}{\ell}+\log n)$ space and $\Oh(n \log n)$ time, and in $\Oh(\frac{n}{\ell}+\log n)$ space and $\Oh(n\log n +\ell \log^2 n)$ time with high probability using a Las-Vegas randomized algorithm.
\end{theorem}
\begin{proof}
	The algorithm first selects anchors $A_1$ and $A_2$ based on a $(\frac15\ell,\frac15\ell)$-partitioning set,
	as described in \cref{sec:synchronized_position} (the partitioning set can be constructed using a deterministic or a randomized procedure). Additional anchors $A_1'$ and $A_2'$ are selected using \cref{proposition:periodic-anchored-sets} with $d=\frac35\ell$ and $\rho=\frac15\ell$.
	Finally, the algorithm runs the procedure of \cref{thm:lcas} with anchors $A_1\cup A'_1$ and $A_2\cup A'_2$ and forwards the obtained result to the output.  
	
	With this selection of anchors, every common substring $T$ of length $|T|\ge \ell$
	is a common anchored substring. 
	Depending on whether $T$ contains a $(\frac35\ell,\frac15\ell)$-run or not, this follows from \cref{proposition:periodic-anchored-sets} and \cref{lemma:findLCSifnotPeriodic}, respectively.
	Consequently, the solution to the LCAS problem is a common substring of length at least $|T|$.
	
	\cref{proposition:periodic-anchored-sets} yields $\Oh(\frac{n}{\ell})$ anchors whereas
	a partitioning set yields $\Oh(\frac{n\log^* n}{\ell})$ or $\Oh(\frac{n}{\ell})$ anchors, depending on whether
	a deterministic or a randomized construction is used. Consequently, the space and time complexity is as stated in 
	the theorem, with the cost dominated by both the partitioning set construction and the algorithm of \cref{thm:lcas}.
\end{proof}
\begin{remark}\label{rem:extra}
	Note that the algorithms of \cref{thm:base_simple} return a longest common substring 
	as long as $\lcs(S_1,S_2)\ge \ell$ (and not just when $\ell \le \lcs(S_1,S_2) \le 2\ell$ as \LCSl requires).
\end{remark}

\subsection{$\Oh(1)$-space algorithm for $\ell=\Omega(n)$}

In \cref{thm:base_simple}, the space usage involves an $\Oh(\log n)$ term, which becomes dominant for very large $\ell$. 
In this section, we design an alternative $\Oh(1)$-space algorithm for $\ell = \Omega(n)$.
Later, in \cref{thm:findLCSwithConstantSpace}, we generalize this algorithm to arbitrary $\ell$, which lets us obtain an analog of \cref{thm:base_simple} with the $\Oh(\log n)$ term removed from the space complexity.

Our main tool is a constant-space pattern matching algorithm.
\begin{lemma}[Galil-Seiferas~\cite{DBLP:journals/jcss/GalilS83}, Crochemore-Perrin~\cite{CP91}]\label{lem:constant-space-pm}
	There exists an $\Oh(1)$-space $\Oh(|P|+|T|)$-time algorithm that, given a read-only pattern $P$ and a read-only text $T$, reports the occurrences of $P$ in $T$ in the left-to-right order.
\end{lemma}

\begin{lemma}\label{lem:smallSpaceAlgorithm}
	The \LCSl problem can be solved deterministically in $\Oh(1)$ space and $\Oh(n)$ time
	for $\ell = \Omega(n)$.
\end{lemma}
\begin{proof}
	We show how to find $\Oh(1)$ anchors such that if $T$ is a common substring of $S_1$ and $S_2$  of length $|T|\ge\ell$, then $T$ is a common anchored substring of $S_1$ and $S_2$.
	
	We first use \cref{proposition:periodic-anchored-sets} with $d=\frac{3}{5}\ell$ and $\rho=\frac{1}{5}\ell$ to generate anchors $A'_1$ and $A'_2$ for the periodic case. The set of these anchors has a size of $\Oh(\frac nd)=\Oh(1)$,
	and, if $T$ contains a $(\frac{3}{5}\ell,\frac{1}{5}\ell)$-run, then $T$ is  a common anchored substring of $S_1$ and $S_2$ with respect to $A'_1,A'_2$.
	
	In order to accommodate the case where $T$ does not contain any $(\frac{3}{5}\ell,\frac{1}{5}\ell)$-run, we construct sets $A_1$ and $A_2$ as follows.
	Consider all the fragments $u_k=S_1[k\frac \ell5\dd (k+3)\frac\ell5-1]$ with boundaries within $[n]$. 
	For each such fragment, add $k\frac{\ell}{5}$ into $A_1$. In addition, use \cref{lem:constant-space-pm} to find all occurrences of $u_k$ in $S_2$, and add all the starting positions of the occurrences to $A_2$,
	unless the number of occurrences exceeds $\frac n{\ell/5}$ (then, $\per(u_k)\le \frac \ell5$).
	The number of fragments $u_k$ is $\Oh(\frac n{\ell/5})=\Oh(1)$, so the sets $A_1$ and $A_2$ contain $\Oh(1)$ elements.
	
	Let  $T=S_1[i_1\dd j_1]=S_2[i_2\dd j_2]$ be  arbitrary occurrences of $T$ in $S_1$ and $S_2$, respectively. Then for $k=\lceil{\frac {i_1}{\ell/5}}\rceil$, the fragment $u_k$ is contained within $S_1[i_1\dd j_1]$. If $\per(u_k)\le \frac15\ell$, then the occurrence of $u_k$ in $T$ can be extended to a $(\frac35\ell,\frac15\ell)$-run in $T$ (and that case has been accommodated using $A'_1$ and $A'_2$). Otherwise, $p_1:=k\frac{\ell}{5}\in A_1$ and all the positions where $u_k$ occurs in $S_2$, including $p_2:=i_2+(k\frac{\ell}{5}-i_1)$,  are in $A_2$. 
	Therefore, $(p_1,p_2)$ is a synchronized pair of anchors and $T$ is a common anchored substring with respect to $A_1,A_2$.

	The number of pairs $(p_1,p_2)\in (A_1\cup A'_1)\times  (A_2\cup A'_2)$ is $\Oh(1)$.
	For each such pair, the algorithm computes $|\LCP(S_1[p_1\dd],S_2[p_2\dd])|+|\LCP^r(S_1[\dd p_1-1],S_2[\dd p_2-1])|$ naively, and returns the common substring corresponding to a maximum among these values. The computation for each pair takes $\Oh(n)$ time. By the argument above, the algorithm finds a common substring of length at least $|T|$ for every common substring $T$ with $|T|\ge \ell$. 
\end{proof}

\section{Time-space tradeoff for the \LCSl problem}\label{sec:knownLenAlgo}
In this section, we show how to use the previous algorithms in order to solve the \LCSl problem
in space $\Oh(s)$, where $s$ is a tradeoff paramater specified on the input.
Our approach relies on the following algorithm which, given a paramter $m\ge \ell$,
reduces a single arbitrary instance of \LCSl to $\Oh(\lceil\frac{n}{m}\rceil^2)$
instances of \LCSl with strings of length $\Oh(m)$.

\begin{algorithm}[H]
	\caption{Self-reduction of \LCSl to many instances on strings of length $\Oh(m)$}\label{alg:reduction}
	\ForEach{$q_1\in [|S_1|]$ s.t. $q_1 \equiv 1 \pmod{m}$ \textrm{and} $q_2\in [|S_2|]$ s.t. $q_2 \equiv 1 \pmod{m}$}{
		Solve \LCSl on $S_1[q_1\dd \min(q_1+3m-1,|S_1|)]$ and $S_2[q_2\dd \min(q_2+3m-1,|S_2|)]$\;
	}
	\Return the longest among the common substrings reported\;
\end{algorithm}

\cref{alg:reduction} clearly reports a common substring of $S_1$ and $S_2$.
Moreover, if $T$ is a common substring of $S_1$ and $S_2$ satisfying $\ell \le |T| \le 2\ell$,
then $T$ is contained in one of the considered pieces $S_1[q_1\dd \min(q_1+3m-1,|S_1|)]$
(the one with $q_1 = 1+m\lfloor\frac{i_1}{m}\rfloor$ if $T=S_1[i_1\dd j_1]$)
and $T=S_2[i_2\dd j_2]$ is contained in one of the considered pieces $S_2[q_2\dd \min(q_2+3m-1,|S_2|)]$
(the one with $q_2 = 1+m\lfloor\frac{i_1}{m}\rfloor$ if $T=S_2[i_2\dd j_2]$).
Thus, the common substring reported by \cref{alg:reduction} satisfies the characterization of the \LCSl problem given in \cref{rem:equiv}.

\begin{theorem}\label{thm:findLCSwithConstantSpace}
	The \LCSl problem can be solved deterministically in $\Oh(1)$
	space and $\Oh(\frac{n^2}{\ell})$ time.
\end{theorem}
\begin{proof}
	We apply the self-reduction of \cref{alg:reduction} with $m=\ell$ to the algorithm of \cref{lem:smallSpaceAlgorithm}.
	The running time is $\Oh((\frac{n}{m})^2 m) = \Oh(\frac{n^2}{m})=\Oh(\frac{n^2}{\ell})$
	and the space complexity is constant.
\end{proof}

This result allows for the aforementioned improvement upon the algorithms of \cref{thm:base_simple}.%
\thmbaseresult*
\begin{proof}
	If $\ell\ge \frac{n}{\log n}$, we use the algorithm of \cref{thm:findLCSwithConstantSpace},
	which costs $\Oh(\frac{n^2}{\ell})=\Oh(n \log n)$ time.
	Otherwise, we use the algorithm of \cref{thm:base_simple}.
	The running time is $\Oh(n\log n + \ell \log^2 n)=\Oh(n \log n)$,
	and the space complexity is $\Oh(\frac{n\log^* n}{\ell}+\log n)=\Oh(\frac{n\log^* n}{\ell})$
	or $\Oh(\frac{n}{\ell}+\log n)=\Oh(\frac{n}{\ell})$, respectively.
\end{proof}

A time-space tradeoff is, in turn, obtained using \cref{alg:reduction} on top of \cref{thm:base_improved}.
\begin{theorem}\label{thm:findLCSwithUser}
	Given a parameter $s\in [1,n]$, the \LCSl problem can be solved deterministically in $\Oh(s)$
	space and $\Oh(\frac{n^2\log n \log^* n}{\ell \cdot s}+n\log n)$ time, and in $\Oh(s)$ space and $\Oh(\frac{n^2\log n}{\ell \cdot s}+n\log n)$ time with high probability using a Las-Vegas randomized algorithm.
\end{theorem}
\begin{proof}
	For a randomized algorithm, we apply the self-reduction of \cref{alg:reduction} with $m=\ell \cdot s$ to the algorithm of \cref{thm:base_improved}.
	The space complexity is $\Oh(\frac{m}{\ell})=\Oh(s)$,
	whereas the running time is $\Oh(n\log n)$ if $m\ge n$ and $\Oh((\frac{n}{m})^2\cdot m\log m) = \Oh(\frac{n^2\log m}{m})= \Oh(\frac{n^2\log n}{\ell \cdot s})$ otherwise.
	
	A deterministic version relies on the algorithm of \cref{thm:findLCSwithConstantSpace}
	for $s < \log^* n$,
	which costs $\Oh(\frac{n^2}{\ell})=\Oh(\frac{n^2\log^* n}{\ell\cdot s})$ time.
	For $s\ge \log^* n$, we apply the self-reduction of \cref{alg:reduction} with $m=\frac{\ell\cdot s}{\log^*n}$ to the algorithm of \cref{thm:base_improved}.
	The space complexity is $\Oh(\frac{m\log^*n}{\ell})=\Oh(s)$,
	whereas the running time is $\Oh(n\log n)$ if $m \ge n$
	and $\Oh((\frac{n}{m})^2\cdot m\log m) = \Oh(\frac{n^2\log m}{m})= \Oh(\frac{n^2\log n \log ^* n}{\ell \cdot s})$ otherwise.
\end{proof}

\section{Time-space tradeoff for the LCS problem}\label{sec:finalAlgo}
In order to solve the LCS problem in time depending on $\lcs(S_1,S_2)$,
we solve \LCSl for exponentially decreasing thresholds $\ell$.

\begin{algorithm}[H]
	\caption{A basic reduction from the LCS problem to the \LCSl problem}\label{alg:findingLength}
	$\ell = n$\;
	\Do{$|T| < \ell$}{
		$\ell = \ell / 2$\;
		$T = \LCSl(S_1,S_2)$\;
	}
	\Return $T$\;
\end{algorithm}

In \cref{alg:findingLength}, as long as $\ell > \lcs(S_1,S_2)$, \LCSl clearly returns a common substring shorter than $\ell$.
In the first iteration when this condition is not satisfied, we have $\ell \le \lcs(S_1,S_2) < 2\ell$,
so \LCSl must return a longest common substring.

If the algorithm of \cref{thm:findLCSwithConstantSpace} is used for \LCSl,
then the space complexity is $\Oh(1)$, and the running time is 
$\Oh(\sum_{i=1}^{\lceil \log \frac{n}{L}\rceil} \frac{n^2}{n/2^i})= \Oh(\sum_{i=1}^{\lceil \log \frac{n}{L}\rceil} n\cdot 2^i) = \Oh(\frac{n^2}{L})$,where $L=\lcs(S_1,S_2)$.%

\thmconstantspaceresult*

The $\Oh(1)$-space solution is still used if the input space restriction is $s=\Oh(\log n)$.
Otherwise, we start with $\ell = \Theta(\frac{n}{s})$ (in the randomized version) or $\ell = \Theta(\frac{n\log^*n}{s})$ (in the deterministic version) and a single call to the algorithm of \cref{thm:base_simple}. This is correct due to \cref{rem:extra},
the space complexity is $\Oh(s)$, and the running time is $\Oh(n \log n)$. 
In subsequent iterations, the procedure of \cref{thm:findLCSwithUser} is used, and its running time is dominated by the first term: $\Oh(\frac{n^2\log n}{\ell \cdot s})$ or $\Oh(\frac{n^2\log n \log^* n}{\ell \cdot s})$, respectively. These values form a geometric progression for exponentially decreasing $\ell$, dominated by the running time of the last iteration: $\Oh(\frac{n^2\log n}{L \cdot s})$ and $\Oh(\frac{n^2\log n \log^* n}{L \cdot s})$, respectively. This analysis yields our main result: 
{\renewcommand\footnote[1]{}\thmmainresult*}

\bibliographystyle{plainurl}
\bibliography{lcs_bibl}
	
\end{document}